\documentclass[a4paper,12pt]{article}
\usepackage{fullpage}
\usepackage[english]{babel}
\usepackage{amsthm}
\usepackage{amsmath}
\usepackage{amssymb,authblk}
\usepackage{t1enc}
\usepackage[utf8]{inputenc}
\usepackage{overpic}
\usepackage{float}
\usepackage{tikz}
\usepackage{cite}
\usepackage{amsrefs}
\usepackage{rotating}
\usetikzlibrary{calc}
\usetikzlibrary{decorations}
\usetikzlibrary{decorations.pathreplacing}
\usepackage{amsrefs}

\newtheorem{thm}{Theorem}[section]
\newtheorem{lemma}[thm]{Lemma}
\newtheorem{prop}[thm]{Proposition}
\newtheorem{corollary}[thm]{Corollary}
\newtheorem{defi}[thm]{Definition}

\newtheorem{claim}[thm]{Claim}

\theoremstyle{remark}

\begin{document}
\title{Strengthening some complexity results on toughness of graphs}

\author[1,2]{Gyula Y Katona\thanks{kiskat@cs.bme.hu}}
\author[1,3]{Kitti Varga\thanks{vkitti@cs.bme.hu}}
\affil[1]{Department of Computer Science and
Information Theory, Budapest University of Technology and Economics, Hungary}
\affil[2]{MTA-ELTE Numerical Analysis and Large Networks Research Group, Hungary}
\affil[3]{HAS Alfr\'ed R\'enyi Institute of Mathematics, Hungary}

\maketitle

\begin{abstract}
 Let $t$ be a positive real number. A graph is called $t$-tough if the removal of any vertex set $S$ that disconnects the graph leaves at most $|S|/t$ components. The toughness of a graph is the largest $t$ for which the graph is $t$-tough.
 
 The main results of this paper are the following. For any positive rational number $t \le 1$ and for any $k \ge 2$ and $r \ge 6$ integers recognizing $t$-tough bipartite graphs is coNP-complete (the case $t=1$ was already known), and this problem remains coNP-complete for $k$-connected bipartite graphs, and so does the problem of recognizing 1-tough r-regular bipartite graphs. To prove these statements we also deal with other related complexity problems on toughness.
\end{abstract}

\section{Introduction}

All graphs considered in this paper are finite, simple and undirected. Let $\omega(G)$ denote the number of components, $\alpha(G)$ the independence number, $\kappa(G)$ the connectivity number and $\delta(G)$ the minimum degree of a graph $G$. For a vertex $v$ of $G$ the degree of $v$ is denoted by $d(v)$. (Using $\omega(G)$ to denote the number of components may be confusing, however, most of the literature on toughness uses this notation.)

The notion of toughness was introduced by Chv\'atal \cite{toughness_intro} to investigate hamiltonicity.

\begin{defi}
 Let $t$ be a real number. A graph $G$ is called {\em $t$-tough} if $|S| \ge t \omega(G-S)$ holds for any vertex set $S \subseteq V(G)$ that disconnects the graph (i.e. for any $S \subseteq V(G)$ with $\omega(G-S)>1$). The \emph{toughness} of $G$, denoted by $\tau(G)$, is the largest $t$ for which G is $t$-tough, taking $\tau(K_n) = \infty$ for all $n \ge 1$.
 
 We say that a cutset $S \subseteq V(G)$ is a {\em tough set} if $\omega(G - S) = |S|/\tau(G)$.
\end{defi}

Clearly, if a graph is Hamiltonian, then it must be 1-tough. However, not every 1-tough graph contains a Hamiltonian cycle: a well-known counterexample is the Petersen graph. On the other hand, Chv\'atal conjectured that there exists a constant $t_0$ such that every $t_0$-tough graph is Hamiltonian~\cite{toughness_intro}. This conjecture is still open, but it is known that, if exists, $t_0$ must be at least 9/4~\cite{9/4}.

The complexity of recognizing $t$-tough graphs has also been in the interest of research. This paper is motivated by two open problems regarding the complexity of recognizing 1-tough 3-connected bipartite graphs and 1-tough 3-regular bipartite graphs.

Let $t$ be an arbitrary positive rational number and consider the following problem.

\medskip
\noindent {\bf \scshape $\boldsymbol{t}$-Tough} \\
\textit{Instance:} a graph $G$. \\
\textit{Question:} is it true that $\tau(G) \ge t$?
\medskip

It is easy to see that for any positive rational number $t$ the problem {\scshape $t$-Tough} is in coNP: a witness is a vertex set $S$ whose removal disconnects the graph and leaves more than $|S|/t$ components. Bauer et al$.$ proved that this problem is coNP-complete~\cite{recognize_toughness} and the problem {\scshape $1$-Tough} remains coNP-complete for at least 3 regular graphs~\cite{recognize_toughness_regular}.

\begin{thm}[\cite{recognize_toughness}] \label{t_tough_conp_complete}
 For any positive rational number $t$ the problem {\scshape $t$-Tough} is coNP-complete.
\end{thm}

\begin{thm}[\cite{recognize_toughness_regular}] \label{1_tough_regular_conp_complete}
 For any fixed integer $r \ge 3$ the problem {\scshape $1$-Tough} is coNP-complete for $r$-regular graphs.
\end{thm}

Although the toughness of any bipartite graph (except for the graphs $K_1$ and $K_2$) is at most one,
the problem {\scshape $1$-Tough} does not become easier for bipartite graphs.

\begin{thm}[\cite{recognize_toughness_bipartite}] \label{1_tough_bipartite_conp_complete}
 The problem {\scshape $1$-Tough} is coNP-complete for bipartite graphs.
\end{thm}

Let $t$ be an arbitrary positive rational number and now consider a variant of the problem {\scshape $t$-Tough}.

\medskip
\noindent {\bf \scshape Exact-$\boldsymbol{t}$-Tough} \\
\textit{Instance:} a graph $G$. \\
\textit{Question:} is it true that $\tau(G)=t$?
\medskip

Extremal problems usually seem not to belong to $\text{NP} \cup \text{coNP}$, therefore a complexity class called DP was introduced by Papadimitriou and Yannakakis~\cite{dp_intro}.

\begin{defi}
 A language $L$ is in the class \emph{DP} if there exist two languages $L_1 \in \text{NP}$ and $L_2 \in \text{coNP}$ such that $L = L_1 \cap L_2$.
 
 A language is called \emph{DP-hard} if all problems in DP can be reduced to it in polynomial time. A language is \emph{DP-complete} if it is in DP and it is DP-hard.
\end{defi}

We mention that $\text{DP} \ne \text{NP} \cap \text{coNP}$ if $\text{NP} \ne \text{coNP}$. Moreover, $\text{NP} \cup \text{coNP} \subseteq \text{DP}$. Now we present some related DP-complete problem.

\medskip
\noindent {\bf \scshape ExactClique} \\
\textit{Instance:} a graph $G$ and a positive rational number $k$. \\
\textit{Question:} is it true that the largest clique of $G$ has size exactly $k$?

\begin{thm}[\cite{dp_intro}] \label{exactclique_dp_complete}
 The problem {\scshape ExactClique} is DP-complete.
\end{thm}

By taking the complement of the graph, we can obtain {\scshape ExactIndependenceNumber} from {\scshape ExactClique}.

\medskip
\noindent {\bf \scshape ExactIndependenceNumber} \\
\textit{Instance:} a graph $G$ and a positive rational number $k$. \\
\textit{Question:} is it true that $\alpha(G)=k$?
\medskip

Since the clique number of a graph is exactly $k$ if and only if the independence number of its complement is exactly $k$, it follows from Theorem~\ref{exactclique_dp_complete} that the problem {\scshape ExactIndependenceNumber} is also DP-complete.

\begin{corollary} \label{exactindepnumber_dp_complete}
 The problem {\scshape ExactIndependenceNumber} is DP-complete.
\end{corollary}

In this paper, first, we prove that {\scshape Exact-$t$-Tough} is DP-complete for any positive rational number $t$, moreover, if $t<1$, then the problem remains DP-complete for bipartite graphs. Note that since the toughness of any bipartite graph (except for $K_1$ and $K_2$) is at most 1, the problem {\scshape Exact-$1$-Tough-Bipartite} is coNP-complete as stated in Theorem~\ref{1_tough_bipartite_conp_complete}.

\begin{thm} \label{exact_t_tough_dp_complete}
 For any positive rational number $t$ the problem {\scshape Exact-$t$-Tough} is DP-complete.
\end{thm}

\begin{thm} \label{exact_t_tough_bipartite_dp_complete}
 For any positive rational number $t < 1$ the problem {\scshape Exact-$t$-Tough} remains DP-complete for bipartite graphs.
\end{thm}

\begin{thm} \label{t_tough_bipartite_conp_complete}
 For any positive rational number $t \le 1$ the problem {\scshape $t$-Tough} remains coNP-complete for bipartite graphs.
\end{thm}

Note that Theorem~\ref{t_tough_bipartite_conp_complete} contains Theorem~\ref{1_tough_bipartite_conp_complete} as a special case.

Our constructions used in the proofs of the above three theorems also provide alternative proofs for Theorems~\ref{t_tough_conp_complete} and~\ref{1_tough_bipartite_conp_complete}.
Furthermore, using the same construction as in the proof of Theorem~\ref{exact_t_tough_bipartite_dp_complete}, we also prove that {\scshape $t$-Tough} remains coNP-complete for $k$-connected bipartite graphs and so does {\scshape $1$-Tough} for $r$-regular bipartite graphs, where $t \le 1$ is an arbitrary rational number and $k \ge 2$ and $r \ge 6$ are integers. Determining the complexity of recognizing $k$-connected bipartite graphs and 1-tough 3-regular bipartite graphs was posed as an open problem in~\cite{recognize_toughness_cubic}. The latter problem remains open along with the problems of recognizing 1-tough 4-regular and 5-regular bipartite graphs.

\begin{thm} \label{t_tough_k_connected_bipartite_conp_complete}
 For any fixed integer $k \ge 2$ and positive rational number $t \le 1$ the problem {\scshape $t$-Tough} remains coNP-complete for $k$-connected bipartite graphs.
\end{thm}

\begin{thm} \label{1_tough_at_least_6_regular_bipartite_conp_complete}
 For any fixed integer $r \ge 6$ the problem {\scshape $1$-Tough} remains coNP-complete for $r$-regular bipartite graphs.
\end{thm}

In order to prove Theorem~\ref{1_tough_at_least_6_regular_bipartite_conp_complete}, we study the problem {\scshape $1/2$-Tough} in the class of $r$-regular graphs: we show that it is coNP-complete if $r \ge 5$ but is in P if $r \le 4$. (Note that the cases $r = 1$ and $r=2$ are trivial.)

\begin{thm} \label{1/2_tough_at_least_5_regular_conp_complete}
 For any fixed integer $r \ge 5$ the problem {\scshape $1/2$-Tough} remains coNP-complete for $r$-regular graphs.
\end{thm}

\begin{thm} \label{less_than_2/3_tough_3_regular_poly}
 For any positive rational number $t < 2/3$ there is a polynomial time algorithm to recognize $t$-tough $3$-regular graphs.
\end{thm}

\begin{thm} \label{less_than_2/3_tough_4_regular_poly}
 There is a polynomial time algorithm to recognize $1/2$-tough 4-regular graphs.
\end{thm}

Note that by Theorem~\ref{1_tough_regular_conp_complete}, recognizing 1-tough 3-regular graphs is coNP-complete. We remark that the toughness of a 3-regular graph (except for $K_4$) is at most $3/2$ and Jackson and Katerinis gave a characterization of cubic 3/2-tough graphs and these graphs can be recognized in polynomial time \cite{3/2_tough_characterization}. Their characterization uses the concept of inflation, which was introduced by Chvátal in \cite{toughness_intro}, but is not presented here.

\begin{thm}[\cite{3/2_tough_characterization}]
 A cubic graph $G$ is $3/2$-tough if and only if $G \simeq K_4$, $G \simeq K_2 \times K_3$, or $G$ is the inflation of a 3-connected cubic graph.
\end{thm}

This paper is structured as follows. After proving some useful lemmas in Section~\ref{section_preliminaries}, we prove Theorem~\ref{exact_t_tough_dp_complete} in Section~\ref{section_general_graphs}. In Section~\ref{section_bipartite_graphs} we prove two theorems about bipartite graphs, Theorems~\ref{exact_t_tough_bipartite_dp_complete} and \ref{t_tough_k_connected_bipartite_conp_complete}. Section~\ref{section_regular_graphs} is about regular graphs, where we prove Theorems~\ref{1_tough_at_least_6_regular_bipartite_conp_complete}--\ref{less_than_2/3_tough_4_regular_poly}.

\section{Preliminaries} \label{section_preliminaries}

In this section we prove some useful lemmas.

\begin{prop} \label{spanning_subgraph_with_toughness_1/2}
 Let $G \ncong K_1, K_2$ be a $1/2$-tough graph. Then there exists a spanning subgraph $H$ of $G$ for which $\tau(H) = 1/2$.
\end{prop}
\begin{proof}
 Let $H$ be a spanning subgraph of $G$ so that $\tau(H) \ge 1/2$ and there exists an edge $e \in E(H)$ for which $\tau(H-e) < 1/2$. (Note that since $\tau(G) \ge 1/2$, such a spanning subgraph $H$ can be obtained by repeatedly deleting some edges of $G$.) Now we show that $\tau(H) \le 1/2$, which implies that $\tau(H) = 1/2$. Let $e \in E(G)$ be an edge for which $\tau(H-e) < 1/2$.
 
 \medskip
 
 \textit{Case 1:} $e$ is a bridge in $H$.
 
 Since $G$ is $1/2$-tough, it is connected. Since $G \ncong K_1, K_2$ and $G$ is connected, the graphs $G$ and $H$ have at least three vertices. Hence, at least one of the endpoint of $e$ is a cut-vertex in $H$, so $\tau(H) \le 1/2$.
 
 \medskip
 
 \textit{Case 2:} $e$ is not a bridge in $H$.
 
 Then there exists a cutset $S$ in $H-e$ for which 
  \[ \omega \big( (H-e) - S \big) > 2|S| \text{.} \]
 
 \medskip
 
 \textit{Case 2.1:} ($e$ is not a bridge in $H$) and $S$ is a cutset in $H$.
 
 Then
 \[ \omega(H-S) \le 2|S| \text{,} \]
 which is only possible if
 \[ \omega(H-S) = 2|S| \qquad \text{and} \qquad \omega \big( (H-e) - S \big) = 2|S| + 1 \text{.} \]
 Therefore, $\tau(H) \le 1/2$.
 
 \medskip
 
 \textit{Case 2.2:} ($e$ is not a bridge in $H$) and $S$ is not a cutset in $H$.
 
 This is only possible if
  \[ \omega \big( (H-e) - S \big) = 2 \text{.} \]
 Hence
  \[ 2 = \omega \big( (H-e) - S \big) > 2|S| \text{,} \]
 i.e. $|S| < 1$, which means that $S=\emptyset$, so $e$ is a bridge $H$, which is a contradiction.
\end{proof}

\begin{prop} \label{obs_below1_v2}
 Let $t \le 1$ be a positive rational number and $G$ a $t$-tough graph. Then 
 \[ \omega(G-S) \le |S|/t \]
 for any proper subset $S$ of $V(G)$.
\end{prop}

\begin{proof}
 If $S$ is a cutset in $G$, then by the definition of toughness $\omega(G-S) \le |S|/t$ holds.

 If $S$ is not a cutset in $G$, then $\omega(G-S) = 1$ since $S \ne V(G)$. On the other hand, $|S|/t \ge 1$ since $S \ne \emptyset$ and $t \le 1$. Therefore, $\omega(G-S) \le |S|/t$ holds in this case as well.
\end{proof}

As is clear from its proof, the above proposition holds even if $S$ is not a cutset. However, it does not hold if $t > 1$ and $S$ is not a cutset: if $t>1$, then the graph cannot contain a cut-vertex; therefore $\omega(G-S) = 1$ for any subset $S$ with $|S| = 1$, while $|S|/t = 1/t < 1$.

\begin{prop} \label{possible_values_of_toughness}
 Let $G$ be a connected noncomplete graph on $n$ vertices. Then $\tau(G)$ is a positive rational number, and if $\tau(G) = a/b$, where $a,b$ are relatively prime positive integers, then $1 \le a,b \le n-1$.
\end{prop}
\begin{proof}
 By definition, 
  \[ \tau(G) = \min_{\substack{S \subseteq V(G) \\ \omega(G-S) \ge 2}} \frac{|S|}{\omega(G-S)} \]
 for a noncomplete graph $G$. Since $G$ is connected and noncomplete, $1 \le |S| \le n-2$ for every $S \subseteq V(G)$ with $\omega(G-S) \ge 2$. Obviously, $\omega(G-S) \ge 2$ and since $G$ is connected, $\omega(G-S) \le n-1$.
\end{proof}

The following is a trivial consequence of Proposition~\ref{possible_values_of_toughness}.

\begin{corollary} \label{toughness_gap}
 Let $G$ and $H$ be two connected noncomplete graphs on $n$ vertices. If $\tau(G) \ne \tau(H)$, then
 \[ \big| \tau(G) - \tau(H) \big| > \frac{1}{n^2} \text{.} \]
\end{corollary}

\begin{claim} \label{exact_t_tough_DP}
 For any positive rational number $t$ the problem {\scshape Exact-$t$-Tough} belongs to DP.
\end{claim}
\begin{proof}
 For any positive rational number $t$,
 \begin{gather*}
  \text{{\scshape Exact-$t$-Tough}} = \{ G \text{ graph} \mid \tau(G) = t \} \\
  = \{ G \text{ graph} \mid \tau(G) \ge t \} \cap \{ G \text{ graph} \mid \tau(G) \le t \} \text{.}
 \end{gather*}
 Let
  \[ L_1 = \{ G \text{ graph} \mid \tau(G) \le t \} \]
 and
  \[ L_2 = \{ G \text{ graph} \mid \tau(G) \ge t \} \text{.} \]
 Notice that $L_2 = \text{{\scshape $t$-Tough}}$ and it is known to be in \text{coNP}: a witness is a vertex set $S \subseteq V(G)$ whose removal disconnects $G$ and leaves more than $|S|/t$ components.
 
 Now we show that $L_1 \in \text{NP}$, i.e$.$ we can express $L_1$ in the form
 \[ L_1 = \{ G \text{ graph} \mid \tau(G) < t + \varepsilon \} \text{,} \]
 which is the complement of a language belonging to coNP. Let $G$ be an arbitrary graph on $n$ vertices. If $G$ is disconnected, then $\tau(G) = 0$, and if $G$ is complete, then $\tau(G) = \infty$, so in both cases $\tau(G) \le t$ if and only if $\tau(G) < t + \varepsilon$ for any positive $\varepsilon$. If $G$ is connected and noncomplete, then from Corollary~\ref{toughness_gap} it follows that $\tau(G) \le t$ if and only if $\tau(G) < t + 1/n^2$. Therefore,
 \[ L_1 = \{ G \text{ graph} \mid \tau(G) \le t \} = \left\{ G \text{ graph} ~ \middle| ~ \tau(G) < t + \frac{1}{|V(G)|^2} \right\} \text{,} \]
 so $L_1 \in \text{NP}$. Hence, we can conclude that $\text{{\scshape Exact-$t$-Tough}} = L_1 \cap L_2 \in \text{DP}$.
\end{proof}

For any positive rational number $t$ let {\scshape Exact-$t$-Tough-Bipartite} denote the problem of determining whether a given bipartite graph has toughness $t$. Since the toughness of a bipartite graph is at most 1 (except for the graphs $K_1$ and $K_2$), we can conclude the following.

\begin{corollary} \label{exact_t_tough_bipartite_DP}
 For any positive rational number $t \le 1$ the problem {\scshape Exact-$t$-Tough-Bipartite} belongs to DP. Moreover, {\scshape Exact-$1$-Tough-Bipartite} belongs to coNP.
\end{corollary}

\section{The complexity of determining the toughness of general graphs, proof of Theorem~\ref{exact_t_tough_dp_complete}} \label{section_general_graphs}

\textbf{Proof of Theorem~\ref{exact_t_tough_dp_complete}.} In Claim \ref{exact_t_tough_DP} we already proved that $\text{{\scshape Exact-$t$-Tough}} \in DP$. To prove {\scshape Exact-$t$-Tough} is DP-hard, we reduce {\scshape ExactIndependenceNumber} (which is DP-complete by Corollary~\ref{exactindepnumber_dp_complete}) to it.

Let $G$ be an arbitrary connected graph on the vertices $v_1, \ldots, v_n$ and let $a,b$ be positive integers such that $t=a/b$. Let $k$ be a positive integer and let $G_k$ be the following graph. For all $i \in [n]$ let
 \[ V_i = \{ v_{i,1}, v_{i,2}, \ldots, v_{i,a} \} \text{,} \]
and let
 \begin{gather*}
  V = \bigcup_{i=1}^n V_i \text{,} \qquad
  U = \bigcup_{i=1}^n \bigcup _{j=1}^{b} u_{i,j} \text{,} \qquad U' = \{ u'_1, \ldots, u'_{(b-1)k} \} \text{,} \qquad W = \{ w_1, \ldots, w_{ak} \} \text{,} \\
  V \left( G_k \right) = V \cup U \cup U' \cup W \text{.}
 \end{gather*}
For all $i \in [n]$ place a clique on $V_i$. For all $i_1,i_2 \in [n]$ if $v_{i_1} v_{i_2} \in E(G)$, then place a complete bipartite graph on $(V_{i_1}; V_{i_2})$. For all $i \in [n]$ and $j \in [b]$ connect $u_{i,j}$ to every vertex of $V_i$. Place a clique on $W$ and connect every vertex of $W$ to every vertex of $V \cup U \cup U'$. See Figure~\ref{figure_exact_t_tough_dp_complete}.

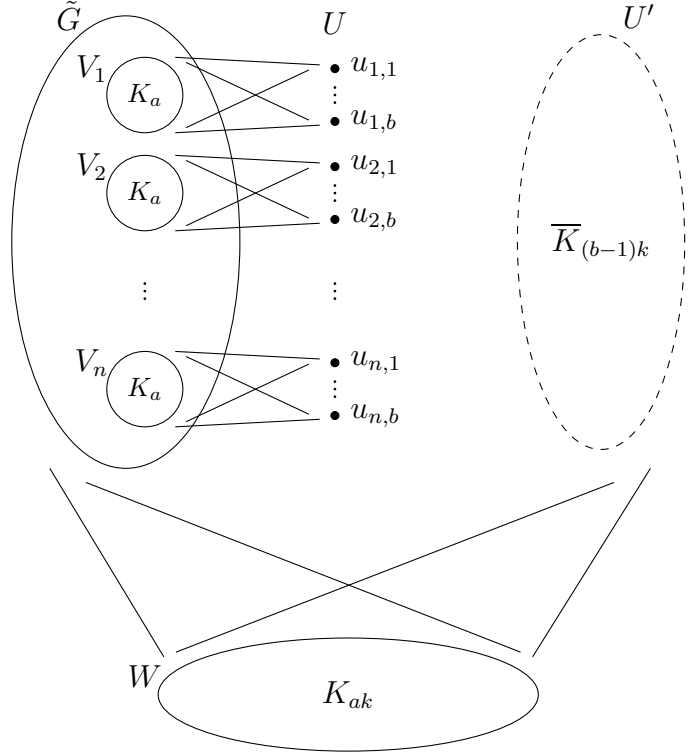
\begin{figure}[H]
 \begin{center}
 \begin{tikzpicture}
  \tikzstyle{vertex}=[draw,circle,fill=black,minimum size=3,inner sep=0]
  
  \draw (-0.25,0) ellipse (1.5 and 3);
  \draw (0,1.95) ellipse (0.5 and 0.5);
  \draw (0,0.65) ellipse (0.5 and 0.5);
  \draw[fill] (0,-0.55) circle (0.3pt);
  \draw[fill] (0,-0.65) circle (0.3pt);
  \draw[fill] (0,-0.75) circle (0.3pt);
  \draw (0,-1.95) ellipse (0.5 and 0.5);
  
  \node at (-1,3) {$\tilde{G}$};
  \node at (-0.7,2.3) {$V_1$};
  \node at (-0.7,1) {$V_2$};
  \node at (-0.7,-1.6) {$V_n$};
  
  \node at (0,1.95) {\footnotesize{$K_a$}};
  \node at (0,0.65) {\footnotesize{$K_a$}};
  \node at (0,-1.95) {\footnotesize{$K_a$}};
  
  \node[vertex] (u11) at (2.5,2.3) [label={[label distance=0]0:$u_{1,1}$}] {};
  \draw[fill] (2.5,2.05) circle (0.3pt);
  \draw[fill] (2.5,1.95) circle (0.3pt);
  \draw[fill] (2.5,1.85) circle (0.3pt);
  \node[vertex] (u1b) at (2.5,1.6) [label={[label distance=0]0:$u_{1,b}$}] {};
  \node[vertex] (u21) at (2.5,1) [label={[label distance=0]0:$u_{2,1}$}] {};
  \draw[fill] (2.5,0.75) circle (0.3pt);
  \draw[fill] (2.5,0.65) circle (0.3pt);
  \draw[fill] (2.5,0.55) circle (0.3pt);
  \node[vertex] (u2b) at (2.5,0.3) [label={[label distance=0]0:$u_{2,b}$}] {};
  \draw[fill] (2.5,-0.55) circle (0.3pt);
  \draw[fill] (2.5,-0.65) circle (0.3pt);
  \draw[fill] (2.5,-0.75) circle (0.3pt);
  \node[vertex] (un1) at (2.5,-1.6) [label={[label distance=0]0:$u_{n,1}$}] {};
  \draw[fill] (2.5,-1.85) circle (0.3pt);
  \draw[fill] (2.5,-1.95) circle (0.3pt);
  \draw[fill] (2.5,-2.05) circle (0.3pt);
  \node[vertex] (unb) at (2.5,-2.3) [label={[label distance=0]0:$u_{n,b}$}] {};
  
  \node at (2.5,2.9) {$U$};
  
  \draw (0.4,2.45) -- (2.3,2.35);
  \draw (0.4,1.45) -- (2.3,1.55);
  \draw ($(0.4,2.45)!0.075!(2.3,1.55)$) -- ($(0.4,2.45)!0.925!(2.3,1.55)$);
  \draw ($(0.4,1.45)!0.075!(2.3,2.35)$) -- ($(0.4,1.45)!0.925!(2.3,2.35)$);
  \draw (0.4,1.15) -- (2.3,1.05);
  \draw (0.4,0.15) -- (2.3,0.25);
  \draw ($(0.4,1.15)!0.075!(2.3,0.25)$) -- ($(0.4,1.15)!0.925!(2.3,0.25)$);
  \draw ($(0.4,0.15)!0.075!(2.3,1.05)$) -- ($(0.4,0.15)!0.925!(2.3,1.05)$);
  \draw (0.4,-2.45) -- (2.3,-2.35);
  \draw (0.4,-1.45) -- (2.3,-1.55);
  \draw ($(0.4,-2.45)!0.075!(2.3,-1.55)$) -- ($(0.4,-2.45)!0.925!(2.3,-1.55)$);
  \draw ($(0.4,-1.45)!0.075!(2.3,-2.35)$) -- ($(0.4,-1.45)!0.925!(2.3,-2.35)$);
  
  \draw (2.675,-6) ellipse (2.5 and 0.75);
  \node at (0,-5.75) {$W$};
  
  \draw[dashed] (6,0) ellipse (1.1 and 2.75);
  \node at (6.5,3) {$U'$};
  
  \node at (2.675,-6) {$K_{a k}$};
  \node at (6,0) {$\overline{K}_{(b-1) k}$};
  
  \coordinate (A) at (0.25,-5.5);
  \coordinate (B) at (5.1,-5.5);
  \coordinate (C) at (-1.25,-3);
  \coordinate (D) at (6.65,-3);
  
  \draw (A) -- (C);
  \draw (B) -- (D);
  \draw ($(A)!0.025!(D)$) -- ($(A)!0.925!(D)$);
  \draw ($(B)!0.025!(C)$) -- ($(B)!0.925!(C)$);
 \end{tikzpicture}
 \caption{The graph $G_k$.} \label{figure_exact_t_tough_dp_complete}
 \end{center}
 \end{figure}

Obviously, $G_k$ can be constructed from $G$ in polynomial time. Now we show that $\alpha(G) = k$ if and only if $\tau(G_k) = t = a/b$, i.e.
\begin{enumerate}
 \item[--] if $\alpha(G) > k$, then 
            \[ \frac{|S|}{\omega(G_k - S)} > t \]
           for any cutset $S$ of $G_k$;
 \item[--] if $\alpha(G) < k$, then there exists a cutset $S_0$ of $G_k$ such that
            \[ \frac{|S_0|}{\omega(G_k - S_0)} < t \text{;} \]
 \item[--] if $\alpha(G) = k$, then 
            \[ \frac{|S|}{\omega(G_k - S)} > t \]
           for any cutset $S$ of $G_k$ and there exists a cutset $S_0$ of $G_k$ such that
            \[ \frac{|S_0|}{\omega(G_k - S_0)}  < t \text{.} \]
\end{enumerate}

Let $S \subseteq V \left( G_k \right)$ be an arbitrary cutset of $G_k$. Since $S$ is a cutset, it must contain $W$. Let
\[ I = \{ i \in [n] \mid V_i \subseteq S \} \text{.} \]
After the removal of $W$, the removal of any vertex of $U \cup U'$ or the removal of only a proper subset of $V_i$ for any $i \in [n]$ does not disconnect anything in the graph.
So consider the cutset
\[ S' = S \setminus \left[ (U \cup U') \cup \left( \bigcup_{i \not\in I} V_i \right) \right] \text{.} \]
In $G_k-S'$ there are two types of components: isolated vertices from $U \cup U'$ and components containing at least one vertex from $V$. There are at most $\alpha(G)$ components of the second type since picking a vertex from each such component forms an independent set of $G[V]$. On the other hand, there are exactly $b|I| + |U'| = b|I| + (b-1)k$ components of the first type. So
\[ |S| \ge |S'| = \sum_{i \in I} |V_i| + |W| = a|I| + ak = a \big( |I|+k \big) \]
and
\[ \omega(G_k - S) = \omega(G_k - S') \le \alpha(G) + b|I| + (b-1)k = b \big( |I|+k \big) + \big( \alpha(G) - k \big) \text{.} \]
Therefore,
\[ \frac{|S|}{\omega(G_k - S)} \ge \frac{|S'|}{\omega(G_k - S')} \ge \frac{a \big( |I|+k \big)}{b \big( |I|+k \big) + \big( \alpha(G) - k \big)} \text{.} \]

Let $\{ v_j \in V(G) \mid j \in J \}$ be an independent set of size $\alpha(G)$ in the graph $G$ for some $J \subseteq [n]$, and consider another cutset
\[ S_0 = \left( \bigcup_{i \not\in J} V_i \right) \cup W \]
in $G_k$.
Then
\[ |S_0| = a \big( n - \alpha(G) \big) + ak = a \big( n - \alpha(G) + k \big) \]
and (similarly as before)
\[ \omega(G_k - S_0) = \alpha(G) + b \big( n - \alpha(G) \big) + (b-1)k = b \big( n - \alpha(G) + k \big) + \big( \alpha(G) - k \big) \text{,} \]
so
\[ \frac{|S_0|}{\omega(G_k - S_0)} = \frac{a \big( n - \alpha(G) + k \big)}{b \big( n - \alpha(G) + k \big) + \big( \alpha(G) - k \big)} \text{.} \]
\medskip

\textit{Case 1:} $\alpha(G) < k$.

 Then
 \begin{gather*} 
  \frac{|S|}{\omega(G_k - S)} \ge \frac{a \big( |I|+k \big)}{b \big( |I|+k \big) + \big( \alpha(G) - k \big)} > \frac{a \big( |I| + k \big)}{b \big( |I| + k \big)} = \frac{a}{b} = t
 \end{gather*}
 holds for every cutset $S$ of $G_k$, which implies that $\tau(G_k) > t$.

 \medskip

\textit{Case 2:} $\alpha(G) = k$.
 
 Then
 \begin{gather*} 
  \frac{|S|}{\omega(G_k - S)} \ge \frac{a \big( |I|+k \big)}{b \big( |I|+k \big) + \big( \alpha(G) - k \big)} = \frac{a \big( |I| + k \big)}{b \big( |I| + k \big)} = \frac{a}{b} = t
 \end{gather*}
 holds for every cutset $S$ of $G_k$, which implies that $\tau(G_k) \ge t$.
 
 On the other hand,
 \[ \tau(G_k) \le \frac{|S_0|}{\omega(G_k - S_0)} = \frac{a \big( n - \alpha(G) + k \big)}{b \big( n - \alpha(G) + k \big) + \big( \alpha(G) - k \big)} = \frac{an}{bn} = \frac{a}{b} = t \text{.} \]
 
 Hence, $\tau(G_k) = t$.
 
\medskip

\textit{Case 3:} $\alpha(G) > k$.

 Then
 \[ \tau(G_k) \le \frac{|S_0|}{\omega(G_k - S_0)} = \frac{a \big( n - \alpha(G) + k \big)}{b \big( n - \alpha(G) + k \big) + \big( \alpha(G) - k \big)} < \frac{a \big( n - \alpha(G) + k \big)}{b \big( n - \alpha(G) + k \big)} = \frac{a}{b} = t \text{.} \]

\medskip
 
This means that $\alpha(G) = k$ if and only if $\tau(G_k) = t = a/b$. \hfill\qed

\bigskip

The construction we used here is a slight modification of the one that Bauer et al$.$ used in \cite{recognize_toughness_v2} for proving that for any rational number $t \ge 1$ recognizing $t$-tough graphs is coNP-complete; in their proof a variant of {\scshape IndependenceNumber} is reduced to the complement of {\scshape $t$-Tough}.

Since in our proof $\alpha(G) > k$ if and only if $\tau(G_k) < t$, we can reduce {\scshape IndependenceNumber} to the complement of {\scshape $t$-Tough}, therefore providing another proof of Theorem~\ref{t_tough_conp_complete}.

\section{The complexity of determining the toughness of bipartite graphs, proofs of Theorems \ref{exact_t_tough_bipartite_dp_complete} and \ref{t_tough_k_connected_bipartite_conp_complete}} \label{section_bipartite_graphs}

Let $G$ be an arbitrary connected graph on the vertices $v_1, \ldots, v_n$ and let $B(G)$ be the following bipartite graph.
Let
 \[ V \big( B(G) \big) = \big\{ v_{i,1}, v_{i,2} \mid i \in [n] \big\} \]
and for all $i,j \in [n]$ if $v_i v_j \in E(G)$, then connect $v_{i,1}$ to $v_{j,2}$ and $v_{i,2}$ to $v_{j,1}$. Also for all $i \in [n]$ connect $v_{i,1}$ to $v_{i,2}$. See Figure~\ref{figure_exact_t_tough_bipartite_dp_complete}.

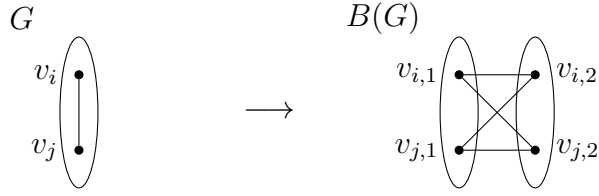
\begin{figure}[H]
\begin{center}
\begin{tikzpicture}
 \tikzstyle{vertex}=[draw,circle,fill=black,minimum size=3,inner sep=0]
 
 \draw (0,0) ellipse (0.25 and 1);
 \draw (5,0) ellipse (0.25 and 1);
 \draw (6,0) ellipse (0.25 and 1);
 
 \node at (2.5,0) {$\longrightarrow$};
 
 \node at (-0.75,1.25) {$G$};
 \node at (4,1.25) {$B(G)$};
 
 \node[vertex] (vi) at (0,0.5) [label={[label distance=2]180: $v_i$}] {};
 \node[vertex] (vj) at (0,-0.5) [label={[label distance=2]180: $v_j$}] {};
 
 \node[vertex] (vi1) at (5,0.5) [label={[label distance=2]180: $v_{i,1}$}] {};
 \node[vertex] (vi2) at (6,0.5) [label={[label distance=2]0: $v_{i,2}$}] {};
 \node[vertex] (vj1) at (5,-0.5) [label={[label distance=2]180: $v_{j,1}$}] {};
 \node[vertex] (vj2) at (6,-0.5) [label={[label distance=2]0: $v_{j,2}$}] {};
 
 \draw (vi) -- (vj);
 \draw (vi1) -- (vj2);
 \draw (vi2) -- (vj1);
 \draw (vi1) -- (vi2);
 \draw (vj1) -- (vj2);

\end{tikzpicture}
\caption{The construction of the bipartite graph $B(G)$.} \label{figure_exact_t_tough_bipartite_dp_complete}
\end{center}
\end{figure}

To prove Theorems~\ref{exact_t_tough_bipartite_dp_complete} and \ref{t_tough_k_connected_bipartite_conp_complete}, first we show how the toughness of $B(G)$ can be computed from the toughness of $G$.

\begin{claim} \label{toughness_of_BG}
 Let $G$ be an arbitrary connected graph. Then $\tau \big( B(G) \big) = \min \big( 2 \tau(G),1 \big)$.
\end{claim}
\begin{proof}
 Let $G$ be an arbitrary graph on the vertices $v_1, \ldots, v_n$ with $\tau(G) = t$.
 
 \bigskip
 
 \textit{Case 1:} $t \le 1/2$.
 
 Let $G' = B(G)$ and let $S_0 \subseteq V(G)$ be an arbitrary tough set in $G$. (Note that since $\tau(G) \le 1/2$, the graph $G$ is noncomplete, therefore it has a tough set.) Consider the vertex set
 \[ S'_0 = \{ v_{i,1}, v_{i,2} \mid v_i \in S_0 \} \text{.} \]
 Clearly, $S'_0$ is a cutset in $G'$ and
 \[ \omega(G' - S'_0) = \omega(G-S_0) = \frac{|S_0|}{t} = \frac{|S'_0|}{2t} \text{,} \]
 so $\tau(G') \le 2t$.
 
 Now we prove that $\tau(G') \ge 2t$, i.e.
  \[ \omega(G' - S') \le \frac{|S'|}{2t} \]
 holds for any cutset $S'$ of $G'$. Therefore, let $S'$ be an arbitrary cutset in $G'$ and let
 \[ S'_1 = \{ v_{i,1} \in S' \mid v_{i,2} \not\in S' \} \cup \{ v_{i,2} \in S' \mid v_{i,1} \not\in S' \} \]
 and
 \[ S'_2 = S' \setminus S'_1 \text{.} \]
 Consider the components of $G' - S'$ which contain either both or none of the vertices $v_{i,1}, v_{i,2}$ for any $i \in [n]$. These components of $G'-S'$ are also components of $G'-S'_2$, so (similarly as before) the number of these components is at most $|S'_2|/2t$.
 The number of the remaining components -- so in which there is at least one vertex without its pair -- can be at most $|S'_1|$, because the pair of the vertex mentioned before must be in $S'_1$. Since $t \le 1/2$,
 \[ \omega(G' - S') \le \frac{|S'_2|}{2t} + |S'_1| \le \frac{|S'_2|}{2t} + \frac{|S'_1|}{2t} = \frac{|S'|}{2t} \text{,} \]
 which implies that $\tau(G') \ge 2t$.
 
 Hence, 
 \[ \tau(G') = 2t = 2\tau(G) = \min \big( 2 \tau(G),1 \big) \text{.} \]
 
 \bigskip
 
 \textit{Case 2:} $t > 1/2$.
 
 By Proposition~\ref{spanning_subgraph_with_toughness_1/2}, there exists a spanning subgraph $H$ with $\tau(H)=1/2$. Then $B(H)$ is a spanning subgraph of $B(G)$, so
 \[ \tau \big( B(G) \big) \ge \tau \big( B(H) \big) \text{,} \]
 and as we saw in Case 1,
 \[ \tau \big( B(H) \big) = 2 \tau(H) = 1 \text{.} \]
 Since $B(G)$ is a bipartite graph, $\tau \big( B(H) \big) \le 1$. Hence,
 \[ \tau \big( B(G) \big) = 1 = \min \big( 2 \tau(G),1 \big) \text{.} \]
\end{proof}

\textbf{Proof of Theorem~\ref{exact_t_tough_bipartite_dp_complete} and alternative proof of Theorem~\ref{1_tough_bipartite_conp_complete}.}
 In Corollary~\ref{exact_t_tough_bipartite_DP} we already proved that if $t \le 1$, then $\text{{\scshape Exact-$t$-Tough-Bipartite}} \in \text{DP}$, moreover, $\text{{\scshape (Exact-)$1$-Tough-Bipartite}} \in \text{coNP}$.
 
 Now we reduce the DP-complete problem {\scshape Exact-$t/2$-Tough} to {\scshape Exact-$t$-Tough-Bipartite} if $t<1$, and the coNP-complete problem {\scshape $1/2$-Tough} to {\scshape (Exact-)$1$-Tough-Bipartite}.
 
 Let $t < 1$ be a positive rational number and let $G$ be an arbitrary connected graph. By Claim~\ref{toughness_of_BG},
 \begin{itemize}
  \item[--] $\tau \big( B(G) \big) = t$ if and only if $\tau(G) = t/2$, and
  \item[--] $\tau \big( B(G) \big) = 1$ if and only if $\tau(G) \ge 1/2$,
 \end{itemize}
 thus the statement of the theorem follows. \hfill\qed

\bigskip

\textbf{Proof of Theorem~\ref{t_tough_bipartite_conp_complete}.}
Since in the above proof $\tau \big( B(G) \big) \ge t$ if and only if $\tau(G) \ge t/2$ for any positive rational number $t \le 1$, we can reduce {\scshape $t/2$-Tough} to {\scshape $t$-Tough-Bipartite}, so the statement of the theorem follows. \hfill\qed

Note that the case $t=1$ was already proved by Kratsch et al$.$ in \cite{recognize_toughness_bipartite}. In their proof the vertices $v_{i,1}$ and $v_{i,2}$ are not connected by an edge, but by a path with two inner vertices. With that construction the original graph is 1-tough if and only if the obtained bipartite graph is exactly 1-tough. However, due to the inner vertices of the paths mentioned before, the constructed bipartite graph has a lot of vertices of degree 2, so these graphs are neither regular (except for cycles) nor 3-connected.

To deal with the problem of determining the complexity of recognizing 3-connected bipartite graphs, we only need one more proposition.

\begin{prop} \label{connectivitynumber_of_BG}
 Let $G$ be an arbitrary graph. Then $\kappa \big( B(G) \big) \ge \kappa(G)$.
\end{prop}
\begin{proof}
 Let $S$ be an arbitrary cutset in $B(G)$. We need to show that $|S| \ge \kappa(G)$.
 
 Let
 \[ W = \big\{ v_{i,1}, v_{i,2} ~ \big| ~ \{v_{i,1}, v_{i,2}\} \cap S = \emptyset \big\} \text{.} \]
 \medskip
 
 \textit{Case 1:} the vertices of $W$ belong to at least two components of $B(G)-S$.
 
 Then
 \[ S' = \{ v_j \in V(G) \mid v_{j,1}, v_{j,2} \notin W \} \]
 is a cutset in $G$: its removal from $G$ disconnects the corresponding vertices of $W$ that belong to different components of $B(G) - S$. Obviously,
 \[ |S| \ge |S'| \ge \kappa(G) \text{.} \]
 
 \medskip
 
 \textit{Case 2:} all vertices of $W$ belong to one component of $B(G)-S$.
 
 Since $S$ is a cutset in $B(G)$, there exists a component $L$ for which $L \cap W = \emptyset$. We can assume that $v_{i,1} \in L$ for some $i \in [n]$. Then $v_{i,2} \in S$ since $L \cap W = \emptyset$. Also, for every $j \in [n]$, if $v_{i} v_{j} \in E(G)$, then either $v_{j,2} \in S$ or $v_{j,2} \in L$, and in the latter case $v_{j,1} \in S$ holds since $L \cap W = \emptyset$. Therefore, 
 \[ |S| \ge d(v_{i,1}) = d(v_i) + 1 \ge \delta(G)+1 > \kappa(G) \text{.} \]
 
 \medskip
 
 Hence, $\kappa \big( B(G) \big) \ge \kappa(G)$.
\end{proof}

\textbf{Proof of Theorem~\ref{t_tough_k_connected_bipartite_conp_complete}.} Let $k \ge 2$ be an integer and $t \le 1$ positive rational number. Applying the proof of Theorem~\ref{t_tough_bipartite_conp_complete} for $k$-connected bipartite graphs, the statement of theorem follows from Proposition~\ref{connectivitynumber_of_BG}. \hfill\qed

\section{On the toughness of regular graphs, proofs of Theorems \ref{1_tough_at_least_6_regular_bipartite_conp_complete}, \ref{1/2_tough_at_least_5_regular_conp_complete}, and \ref{less_than_2/3_tough_3_regular_poly}} \label{section_regular_graphs}

For any positive rational number $t$ and positive integer $r$ let {\scshape $t$-Tough-$r$-Regular} denote the problem of determining whether a given $r$-regular graph is $t$-tough, and let {\scshape $t$-Tough-$r$-Regular-Bipartite} denote the same problem for bipartite graphs.

For any odd number $r \ge 5$ let $H_r$ be the complement of the graph whose vertex set is
 \[ V = \{ w, u_1, \ldots, u_{r+1} \} \]
 and whose edge set is
 \[ E = \left( \bigcup_{i=1}^{\frac{r-1}{2}} \{ u_{i}, u_{r-i+2} \} \right) \cup \{ w, u_{(r+1)/2} \} \cup \{ w, u_{(r+3)/2} \} \text{.} \]
For any even number $r \ge 6$ let $H_r$ be a bipartite graph with color classes
 \[ A = \{ w_a, a_1, \ldots, a_{r-1} \} \qquad \text{and} \qquad B = \{ w_b, b_1, \ldots, b_{r-1} \} \text{,} \]
which can be obtained from the complete bipartite graph by removing the edge $\{ w_a, w_b \}$. (See the graphs $\overline{H}_5$, $H_5$ and $H_6$ in Figure~\ref{figure_H_r}.)

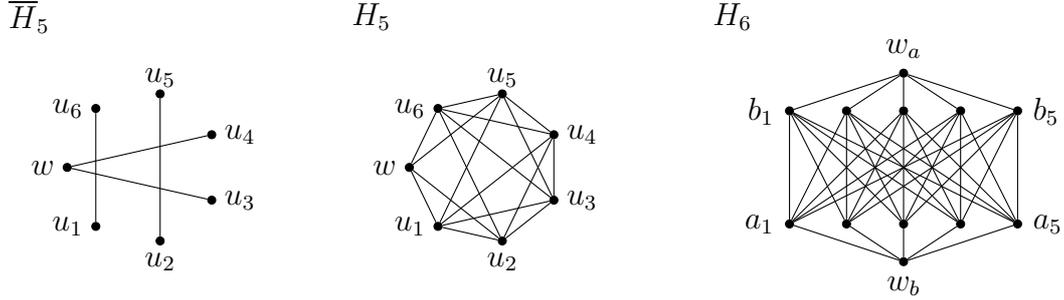
\begin{figure}[H]
\begin{center}
\begin{tikzpicture}
 \tikzstyle{vertex}=[draw,circle,fill=black,minimum size=3,inner sep=0]
 
 \node at (-1.5,2) {$\overline{H}_5$};
 
 \node[vertex] (w) at (180:1) [label={[label distance=-1]180: $w$}] {};
 \node[vertex] (u1) at (180+1*360/7:1) [label={[label distance=-1]180: $u_1$}] {};
 \node[vertex] (u2) at (180+2*360/7:1) [label={[label distance=-1.5]270: $u_2$}] {};
 \node[vertex] (u3) at (180+3*360/7:1) [label={[label distance=-1]0: $u_3$}] {};
 \node[vertex] (u4) at (180+4*360/7:1) [label={[label distance=-1]0: $u_4$}] {};
 \node[vertex] (u5) at (180+5*360/7:1) [label={[label distance=-3]90: $u_5$}] {};
 \node[vertex] (u6) at (180+6*360/7:1) [label={[label distance=-1]180: $u_6$}] {};
 
 \draw (w) -- (u3);
 \draw (w) -- (u4);
 \draw (u1) -- (u6);
 \draw (u2) -- (u5);
 
 \begin{scope}[shift={(4.5,0)}]
 \node at (-1.5,2) {$H_5$};
 
 \node[vertex] (w) at (180:1) [label={[label distance=-1]180: $w$}] {};
 \node[vertex] (u1) at (180+1*360/7:1) [label={[label distance=-1]180: $u_1$}] {};
 \node[vertex] (u2) at (180+2*360/7:1) [label={[label distance=-1.5]270: $u_2$}] {};
 \node[vertex] (u3) at (180+3*360/7:1) [label={[label distance=-1]0: $u_3$}] {};
 \node[vertex] (u4) at (180+4*360/7:1) [label={[label distance=-1]0: $u_4$}] {};
 \node[vertex] (u5) at (180+5*360/7:1) [label={[label distance=-3]90: $u_5$}] {};
 \node[vertex] (u6) at (180+6*360/7:1) [label={[label distance=-1]180: $u_6$}] {};
 
 \draw (w) -- (u1);
 \draw (w) -- (u2);
 \draw (w) -- (u5);
 \draw (w) -- (u6);
 \draw (u1) -- (u2);
 \draw (u1) -- (u3);
 \draw (u1) -- (u4);
 \draw (u1) -- (u5);
 \draw (u2) -- (u3);
 \draw (u2) -- (u4);
 \draw (u2) -- (u6);
 \draw (u3) -- (u4);
 \draw (u3) -- (u5);
 \draw (u3) -- (u6);
 \draw (u4) -- (u5);
 \draw (u4) -- (u6);
 \draw (u5) -- (u6);
 \end{scope}
 
 \begin{scope}[shift={(10,0)}]
 \node at (-2.25,2) {$H_6$};
 
 \node[vertex] (wa) at (0,1.25) [label={above: $w_a$}] {};
 \node[vertex] (wb) at (0,-1.25) [label={below: $w_b$}] {};
 \node[vertex] (a1) at (-1.5,-0.75) [label={left: $a_1$}] {};
 \node[vertex] (a2) at (-0.75,-0.75) {};
 \node[vertex] (a3) at (0,-0.75) {};
 \node[vertex] (a4) at (0.75,-0.75) {};
 \node[vertex] (a5) at (1.5,-0.75) [label={right: $a_5$}] {};
 \node[vertex] (b1) at (-1.5,0.75) [label={left: $b_1$}] {};
 \node[vertex] (b2) at (-0.75,0.75) {};
 \node[vertex] (b3) at (0,0.75) {};
 \node[vertex] (b4) at (0.75,0.75) {};
 \node[vertex] (b5) at (1.5,0.75) [label={right: $b_5$}] {};
 
 \draw (wa) -- (b1);
 \draw (wa) -- (b2);
 \draw (wa) -- (b3);
 \draw (wa) -- (b4);
 \draw (wa) -- (b5);
 \draw (wb) -- (a1);
 \draw (wb) -- (a2);
 \draw (wb) -- (a3);
 \draw (wb) -- (a4);
 \draw (wb) -- (a5);
 \draw (a1) -- (b1);
 \draw (a1) -- (b2);
 \draw (a1) -- (b3);
 \draw (a1) -- (b4);
 \draw (a1) -- (b5);
 \draw (a2) -- (b1);
 \draw (a2) -- (b2);
 \draw (a2) -- (b3);
 \draw (a2) -- (b4);
 \draw (a2) -- (b5);
 \draw (a3) -- (b1);
 \draw (a3) -- (b2);
 \draw (a3) -- (b3);
 \draw (a3) -- (b4);
 \draw (a3) -- (b5);
 \draw (a4) -- (b1);
 \draw (a4) -- (b2);
 \draw (a4) -- (b3);
 \draw (a4) -- (b4);
 \draw (a4) -- (b5);
 \draw (a5) -- (b1);
 \draw (a5) -- (b2);
 \draw (a5) -- (b3);
 \draw (a5) -- (b4);
 \draw (a5) -- (b5);
 \end{scope}

\end{tikzpicture}
\caption{The graphs $\overline{H}_5$, $H_5$ and $H_6$.} \label{figure_H_r}
\end{center}
\end{figure}
 
\begin{claim} \label{toughness_of_H_r}
 For any integer $r \ge 5$, $\tau(H_r) \ge 1$.
\end{claim}

\begin{proof}
 There is a Hamiltonian cycle in $H_r$, namely
  \[ w u_1 u_2 \ldots u_{r+1} w \]
 if $r$ is odd, and
  \[ w_a b_1 a_1 w_b a_2 b_2 a_3 b_3 \ldots a_{r-1} b_{r-1} w_a \]
 if $r$ is even,
 so $\tau(H_r) \ge 1$.
\end{proof}

\begin{lemma} \label{1/2_tough_odd_regular_conp_complete}
 For any fixed odd number $r \ge 5$ the problem {\scshape $1/2$-Tough} is coNP-complete for $r$-regular graphs.
\end{lemma}

\begin{proof}
 Obviously, $\text{\scshape $1/2$-Tough-$r$-Regular} \in \text{coNP}$. To prove that it is coNP-hard we reduce {\scshape $1$-Tough-$(r-1)$-Regular} (which is coNP-complete by Theorem~\ref{1_tough_regular_conp_complete}) to it.
 
 Let $G$ be an arbitrary connected ($r-1$)-regular graph on the vertices $v_1, \ldots, v_n$ and let $G'$ be defined as follows. For all $i \in [n]$ let
 \[ V_i = \{ w^i, u_1^i, \ldots, u_{r+1}^i \} \]
 and place the graph $H_r$ on the vertices of $V_i$ and also connect $v_i$ to $w^i$, see Figure~\ref{figure_1/2_tough_odd_regular_conp_complete}. It is easy to see that $G'$ is $r$-regular and can be constructed from $G$ in polynomial time. Now we prove that $G$ is 1-tough if and only if $G'$ is $1/2$-tough.
 
 \begin{figure}[H]
 \begin{center}
 \begin{tikzpicture}
  \tikzstyle{vertex}=[draw,circle,fill=black,minimum size=4,inner sep=0]
  
  \draw (0,0) ellipse (0.4 and 1.75);
  \node at (-1,2) {$G$};
  
  \node[vertex] (v1) at (0,1.375) [label={[xshift=-15pt, yshift=-10pt] $v_1$}] {};
  \draw[fill] (0,0.1) circle (0.3pt);
  \draw[fill] (0,0) circle (0.3pt);
  \draw[fill] (0,-0.1) circle (0.3pt);
  \node[vertex] (vn) at (0,-1.375) [label={[xshift=-15pt, yshift=-10pt] $v_n$}] {};
  
  \begin{scope}[shift={(2,1.375)}, scale=0.75]  
  \node[vertex] (w) at (180:1) [label={[xshift=-4pt, yshift=-2pt] $w^1$}] {};
  \node[vertex] (u1) at (180+1*360/7:1) [label={[xshift=0pt, yshift=-20pt] $u_1^1$}] {};
  \node[vertex] (u2) at (180+2*360/7:1) [label={[xshift=0pt, yshift=-20pt] $u_2^1$}] {};
  \coordinate (u3) at (180+3*360/7:1);
  \draw[fill] (1,0.2) circle (0.3pt);
  \draw[fill] (1,0) circle (0.3pt);
  \draw[fill] (1,-0.2) circle (0.3pt);
  \coordinate (u4) at (180+4*360/7:1);
  \node[vertex] (u5) at (180+5*360/7:1) [label={[xshift=0pt, yshift=-3pt] $u_r^1$}] {};
  \node[vertex] (u6) at (180+6*360/7:1) [label={[xshift=0pt, yshift=-3pt] $u_{r+1}^1$}] {};
  
  \draw (w) -- (u1);
  \draw (w) -- (u2);
  \draw (w) -- (u5);
  \draw (w) -- (u6);
  \draw (u1) -- (u2);
  \draw (u1) -- ($(u1)!0.85!(u3)$);
  \draw (u1) -- ($(u1)!0.85!(u4)$);
  \draw (u1) -- (u5);
  \draw (u2) -- ($(u2)!0.85!(u3)$);
  \draw (u2) -- ($(u2)!0.85!(u4)$);
  \draw (u2) -- (u6);
  \draw (u5) -- ($(u5)!0.85!(u3)$);
  \draw (u5) -- ($(u5)!0.85!(u4)$);
  \draw (u5) -- (u6);
  \draw (u6) -- ($(u6)!0.85!(u3)$);  
  \draw (u6) -- ($(u6)!0.85!(u4)$);
  
  \draw (w) -- (v1);
  
  \node at (1.75,0.5) {$H_r$};
  \end{scope}
  
  \begin{scope}[shift={(2,-1.375)}, scale=0.75]  
  \node[vertex] (w) at (180:1) [label={[xshift=-4pt, yshift=-2pt] $w^n$}] {};
  \node[vertex] (u1) at (180+1*360/7:1) [label={[xshift=0pt, yshift=-20pt] $u_1^n$}] {};
  \node[vertex] (u2) at (180+2*360/7:1) [label={[xshift=0pt, yshift=-20pt] $u_2^n$}] {};
  \coordinate (u3) at (180+3*360/7:1);
  \draw[fill] (1,0.2) circle (0.3pt);
  \draw[fill] (1,0) circle (0.3pt);
  \draw[fill] (1,-0.2) circle (0.3pt);
  \coordinate (u4) at (180+4*360/7:1);
  \node[vertex] (u5) at (180+5*360/7:1) [label={[xshift=0pt, yshift=-3pt] $u_r^n$}] {};
  \node[vertex] (u6) at (180+6*360/7:1) [label={[xshift=0pt, yshift=-3pt] $u_{r+1}^n$}] {};
  
  \draw (w) -- (u1);
  \draw (w) -- (u2);
  \draw (w) -- (u5);
  \draw (w) -- (u6);
  \draw (u1) -- (u2);
  \draw (u1) -- ($(u1)!0.85!(u3)$);
  \draw (u1) -- ($(u1)!0.85!(u4)$);
  \draw (u1) -- (u5);
  \draw (u2) -- ($(u2)!0.85!(u3)$);
  \draw (u2) -- ($(u2)!0.85!(u4)$);
  \draw (u2) -- (u6);
  \draw (u5) -- ($(u5)!0.85!(u3)$);
  \draw (u5) -- ($(u5)!0.85!(u4)$);
  \draw (u5) -- (u6);
  \draw (u6) -- ($(u6)!0.85!(u3)$);  
  \draw (u6) -- ($(u6)!0.85!(u4)$);
  
  \draw (w) -- (vn);
  
  \node at (1.75,0.5) {$H_r$};
  \end{scope}
 \end{tikzpicture}
 \caption{The graph $G'$ constructed in the proof of Lemma~\ref{1/2_tough_odd_regular_conp_complete}.} \label{figure_1/2_tough_odd_regular_conp_complete}
 \end{center}
 \end{figure}
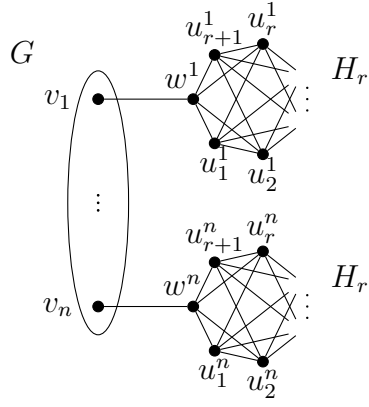
 
 If $G$ is not 1-tough, then there exists a cutset $S \subseteq V(G)$ satisfying $\omega(G-S) > |S|$. Then $S$ is also a cutset in $G'$ and 
 \[ \omega(G'-S) = \omega(G-S) + |S| > 2|S| \text{,} \]
 so $\tau(G') < 1/2$.
 
 Now assume that $G$ is 1-tough. Let $S \subseteq V(G')$ be an arbitrary cutset in $G'$, and let $S_0 = V(G) \cap S$ and $S_i = V_i \cap S$ for all $i \in [n]$.
 Using these notations it is clear that
 \[ S = S_0 \cup \left( \bigcup_{i=1}^n S_i \right) \]
 and
 \[ \omega(G' - S) \le \omega \big( G-S_0 \big) + |S_0| + \sum_{i=1}^n \omega(H_r^i-S_i) \text{,} \]
 where $H_r^i$ denotes the $i$-th copy of $H_r$, i.e. the graph on the vertex set $V_i$ for all $i \in [n]$.
 Since $G$ is 1-tough and by Claim~\ref{toughness_of_H_r}, so is $H_r$, it follows from Proposition~\ref{obs_below1_v2} that
 \[ \omega(G-S_0) \le |S_0| \]
 and
 \[ \omega(H_r^i-S_i) \le |S_i| \text{.} \]
 Therefore,
 \[ \omega(G' - S) \le |S_0| + |S_0| + \sum_{i=1}^n |S_i| \le 2|S| \text{,} \]
 so $\tau(G') \ge 1/2$.
\end{proof}

\begin{lemma} \label{1/2_tough_even_regular_conp_complete}
 For any fixed even number $r \ge 6$ the problem {\scshape $1/2$-Tough} is coNP-complete for $r$-regular graphs.
\end{lemma}

\begin{proof}
 Obviously, $\text{\scshape $1/2$-Tough-$r$-Regular} \in \text{coNP}$. To prove that it is coNP-hard we reduce {\scshape $1$-Tough-$(r-2)$-Regular} (which is coNP-complete by Theorem~\ref{1_tough_regular_conp_complete}) to it.
 
 Let $G$ be an arbitrary connected ($r-2$)-regular graph on the vertices $v_1, \ldots, v_n$ and let $G'$ be defined as follows. For all $i \in [n]$ let
 \[ A_i = \{ w_a^i, a_1^i, \ldots, a_{r-1}^i \} \text{,} \qquad B_i = \{ w_b^i, b_1^i, \ldots, b_{r-1}^i \} \]
 and place the graph $H_r$ on the color classes $A_i$ and $B_i$ and also connect $v_i$ to $w_a^i$ and $w_b^i$, see Figure~\ref{figure_1/2_tough_even_regular_conp_complete}. It is easy to see that $G'$ is $r$-regular and can be constructed from $G$ in polynomial time.
 
 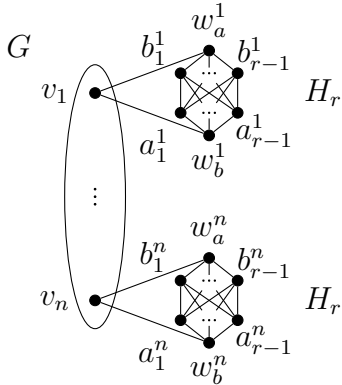
\begin{figure}[H]
 \begin{center}
 \begin{tikzpicture}
  \tikzstyle{vertex}=[draw,circle,fill=black,minimum size=4,inner sep=0]
  
  \draw (0,0) ellipse (0.4 and 1.75);
  \node at (-1,2) {$G$};
  
  \node[vertex] (v1) at (0,1.375) [label={[xshift=-15pt, yshift=-10pt] $v_1$}] {};
  \draw[fill] (0,0.1) circle (0.3pt);
  \draw[fill] (0,0) circle (0.3pt);
  \draw[fill] (0,-0.1) circle (0.3pt);
  \node[vertex] (vn) at (0,-1.375) [label={[xshift=-15pt, yshift=-10pt] $v_n$}] {};
  
  \begin{scope}[shift={(1.5,1.375)}, scale=0.75]  
  \node[vertex] (wa) at (0,0.75) [label={[xshift=0pt, yshift=-1pt] $w_a^1$}] {};
  \node[vertex] (b1) at (-0.5,0.35) [label={[xshift=-10pt, yshift=-3pt] $b_1^1$}] {};
  \coordinate (b2) at (0,0.35);
  \draw[fill] (-0.1,0.35) circle (0.3pt);
  \draw[fill] (0,0.35) circle (0.3pt);
  \draw[fill] (0.1,0.35) circle (0.3pt);
  \node[vertex] (b3) at (0.5,0.35) [label={[xshift=10pt, yshift=-6pt] $b_{r-1}^1$}] {};
  \node[vertex] (a1) at (-0.5,-0.35) [label={[xshift=-10pt, yshift=-25pt] $a_1^1$}] {};
  \coordinate (a2) at (0,-0.35);
  \draw[fill] (-0.1,-0.35) circle (0.3pt);
  \draw[fill] (0,-0.35) circle (0.3pt);
  \draw[fill] (0.1,-0.35) circle (0.3pt);
  \node[vertex] (a3) at (0.5,-0.35) [label={[xshift=10pt, yshift=-20pt] $a_{r-1}^1$}] {};
  \node[vertex] (wb) at (0,-0.75) [label={[xshift=0pt, yshift=-22pt] $w_b^1$}] {};
  
  \draw (wa) -- (b1);
  \draw (wa) -- ($(wa)!0.75!(b2)$);
  \draw (wa) -- (b3);
  \draw (wb) -- (a1);
  \draw (wb) -- ($(wb)!0.75!(a2)$);
  \draw (wb) -- (a3);
  \draw (a1) -- (b1);
  \draw (a1) -- ($(a1)!0.75!(b2)$);
  \draw (a1) -- (b3);
  \draw (a3) -- (b1);
  \draw (a3) -- ($(a3)!0.75!(b2)$);
  \draw (a3) -- (b3);
  \draw (b1) -- ($(b1)!0.75!(a2)$);
  \draw (b3) -- ($(b3)!0.75!(a2)$);
  
  \draw (wa) -- (v1) -- (wb);
  
  \node at (2,0) {$H_r$};
  \end{scope}
  
  \begin{scope}[shift={(1.5,-1.375)}, scale=0.75]  
  \node[vertex] (wa) at (0,0.75) [label={[xshift=0pt, yshift=-1pt] $w_a^n$}] {};
  \node[vertex] (b1) at (-0.5,0.35) [label={[xshift=-10pt, yshift=-3pt] $b_1^n$}] {};
  \coordinate (b2) at (0,0.35);
  \draw[fill] (-0.1,0.35) circle (0.3pt);
  \draw[fill] (0,0.35) circle (0.3pt);
  \draw[fill] (0.1,0.35) circle (0.3pt);
  \node[vertex] (b3) at (0.5,0.35) [label={[xshift=10pt, yshift=-6pt] $b_{r-1}^n$}] {};
  \node[vertex] (a1) at (-0.5,-0.35) [label={[xshift=-10pt, yshift=-25pt] $a_1^n$}] {};
  \coordinate (a2) at (0,-0.35);
  \draw[fill] (-0.1,-0.35) circle (0.3pt);
  \draw[fill] (0,-0.35) circle (0.3pt);
  \draw[fill] (0.1,-0.35) circle (0.3pt);
  \node[vertex] (a3) at (0.5,-0.35) [label={[xshift=10pt, yshift=-20pt] $a_{r-1}^n$}] {};
  \node[vertex] (wb) at (0,-0.75) [label={[xshift=0pt, yshift=-22pt] $w_b^n$}] {};
  
  \draw (wa) -- (b1);
  \draw (wa) -- ($(wa)!0.75!(b2)$);
  \draw (wa) -- (b3);
  \draw (wb) -- (a1);
  \draw (wb) -- ($(wb)!0.75!(a2)$);
  \draw (wb) -- (a3);
  \draw (a1) -- (b1);
  \draw (a1) -- ($(a1)!0.75!(b2)$);
  \draw (a1) -- (b3);
  \draw (a3) -- (b1);
  \draw (a3) -- ($(a3)!0.75!(b2)$);
  \draw (a3) -- (b3);
  \draw (b1) -- ($(b1)!0.75!(a2)$);
  \draw (b3) -- ($(b3)!0.75!(a2)$);
  
  \draw (wa) -- (vn) -- (wb);
  
  \node at (2,0) {$H_r$};
  \end{scope}
 \end{tikzpicture}
 \caption{The graph $G'$ constructed in the proof of Lemma~\ref{1/2_tough_even_regular_conp_complete}.} \label{figure_1/2_tough_even_regular_conp_complete}
 \end{center}
 \end{figure}
 
 Similarly as in the proof of Lemma \ref{1/2_tough_odd_regular_conp_complete}, it can be shown that $G$ is 1-tough if and only if $G'$ is $1/2$-tough.
\end{proof}

\textbf{Proof of Theorem~\ref{1/2_tough_at_least_5_regular_conp_complete}.}
The theorem directly follows from Lemmas~\ref{1/2_tough_odd_regular_conp_complete} and \ref{1/2_tough_even_regular_conp_complete}. \hfill\qed

\bigskip

Using this result, we can prove Theorem~\ref{1_tough_at_least_6_regular_bipartite_conp_complete}.

\bigskip

\textbf{Proof of Theorem~\ref{1_tough_at_least_6_regular_bipartite_conp_complete}.}
 Obviously, $\text{\scshape $1$-Tough-$r$-Regular-Bipartite} \in \text{coNP}$. To prove that it is coNP-hard we reduce {\scshape $1/2$-Tough-$(r-1)$-Regular} (which is coNP-complete by Theorem~\ref{1/2_tough_at_least_5_regular_conp_complete}) to it.
 
 Let $G$ be an arbitrary connected $(r-1)$-regular graph and let $B(G)$ denote the bipartite graph defined at the beginning of Section~\ref{section_bipartite_graphs}. Then $B(G)$ is $r$-regular and by Claim \ref{toughness_of_BG}, the graph $G$ is $1/2$-tough if and only if $B(G)$ is 1-tough. \hfill\qed

\bigskip

For any $r \in \{ 3,4,5 \}$ the problem of determining the complexity of {\scshape $1$-Tough-$r$-Regular-Bipartite} remains open. The reason why our construction does not work in these cases is that we can decide in polynomial time whether an at most 4 regular graph is $1/2$-tough, which we prove in the rest of this paper.

\begin{lemma} \label{less_than_2/3_tough_3_regular_equivalent_characterization}
 For any connected 3-regular graph $G$, the following are equivalent.
 \begin{enumerate}
  \item[$(1)$] There is a cut-vertex in $G$.
  \item[$(2)$] $\tau(G) \le 1/2$.
  \item[$(3)$] $\tau(G) < 2/3$.
 \end{enumerate}
\end{lemma}

\begin{proof} $\empty$

 $(1) \Longrightarrow (2):$ Trivial.
 
 \medskip
 
 $(2) \Longrightarrow (3):$ Trivial.
 
 \medskip
 
 $(3) \Longrightarrow (1):$ If $\tau(G) < 2/3$, then there exists a cutset $S \subseteq V(G)$ satisfying 
 \[ \omega(G-S) > \frac{3}{2} |S| \text{.} \]
 Hence there must exist a component of $G-S$ that has exactly one neighbor in $S$: since $G$ is connected, every component has at least one neighbor in $S$, and if every component of $G-S$ had at least two neighbors in $S$, then the number of edges going into $S$ would be at least $2 \omega(G-S) > 3|S|$, which would contradict the 3-regularity of $G$. Obviously, this neighbor in $S$ is a cut-vertex in $G$.
\end{proof}

\textbf{Proof of Theorem \ref{less_than_2/3_tough_3_regular_poly}.}
 Let $G$ be an arbitrary connected 3-regular graph. First check whether $G$ contains a cut-vertex. By Lemma~\ref{less_than_2/3_tough_3_regular_equivalent_characterization}, if it does not, then $\tau(G) \ge 2/3$, but if it does, then $\tau(G) \le 1/2$. We prove that in the latter case either $\tau(G) = 1/3$ or $\tau(G) = 1/2$, and we can also decide in polynomial time which one holds.
 
 Since $G$ is 3-regular, $\omega(G-S) \le 3|S|$ holds for any cutset $S$ of $G$, so $\tau(G) \ge 1/3$. Now we show that if $\tau(G) < 1/2$, then $\tau(G) \le 1/3$.
 
 
 
 So assume that $\tau(G) < 1/2$ and let $S$ be a tough set of $G$ and let $k = \omega(G-S)$. Then $k > 2|S|$. Contract the components of $G-S$ into single vertices $u_1, \ldots, u_k$ while keeping the multiple edges and let $H$ denote the obtained multigraph. Since $G$ is connected, $d(u_i) \ge 1$ holds for any $i \in [k]$, so
 \[ k = \big| \{i \in [k] : d(u_i) = 1 \} \big| + \big| \{i \in [k] : d(u_i) \ge 2 \} \big| \text{.} \]
 Since $G$ is 3-regular,
 \begin{gather*}
  3|S| \ge \sum_{i=1}^{k} d(u_i) \ge \big| \{i \in [k] : d(u_i) = 1 \} \big| + 2 \cdot \big| \{i \in [k] : d(u_i) \ge 2 \} \big| \\
  = k + \big| \{i \in [k] : d(u_i) \ge 2 \} \big| > 2|S| + \big| \{i \in [k] : d(u_i) \ge 2 \} \big|\text{,}
 \end{gather*}
 so
 \[ |S| > \big| \{i \in [k] : d(u_i) \ge 2 \} \big| \text{.} \]
 Therefore,
 \[ \big| \{i \in [k] : d(u_i) = 1 \} \big| = k - \big| \{i \in [k] : d(u_i) \ge 2 \} \big| > 2|S| - |S| = |S| \text{,} \]
 which means that there exists a vertex in $S$ having at least two neighbors in $\{ u_1, \ldots, u_k \}$ of degree 1. Then the removal of this vertex leaves at least three components (and note that since $G$ is 3-regular, it cannot leave more than three components), so $\tau(G) \le 1/3$.
 
 From this it also follows that $\tau(G) = 1/3$ if and only if there exists a cut-vertex whose removal leaves three components.
 
 To summarize, it can be decided in polynomial time whether a connected 3-regular graph is $2/3$-tough, and if it is not, than its toughness is either 1/3 or 1/2. In both cases the graph contains at least one cut-vertex, and if the removal of any of them leaves (at least) three components, than the toughness of the graph is 1/3, otherwise it is 1/2. \hfill\qed
 
 \bigskip

\begin{claim} \label{toughness_of_connected_4_regular_graphs}
 The toughness of any connected 4-regular graph is at least $1/2$.
\end{claim}
\begin{proof}
 Let $G$ be a  connected 4-regular graph and let $S$ be an arbitrary cutset in $G$ and $L$ be a component of $G-S$. Since every vertex has degree 4 in $G$, the number of edges between $S$ and $L$ is even (more precisely, it is equal to the sum of the degrees in $G$ of the vertices of $L$ minus two times the number of edges induced by $L$). Since $G$ is connected, the number of these edges is at least two. On the other hand, since $G$ is 4-regular, there are at most $4|S|$ edges between $S$ and $L$. Therefore $\omega(G-S) \le 2|S|$, which means that $\tau(G) \ge 1/2$.
\end{proof}
\textbf{Proof of Theorem \ref{less_than_2/3_tough_4_regular_poly}.} It directly follows from Claim~\ref{toughness_of_connected_4_regular_graphs}. \hfill\qed

\section{Acknowledgment}
The research of the first author was supported by National Research,
Development and Innovation Office NKFIH, K-116769 and K-124171, by the
National Research, Development and Innovation Fund
(TUDFO/51757/2019-ITM, Thematic Excellence Program), and by the Higher
Education Excellence Program of the Ministry of Human Capacities in
the frame of Artificial Intelligence research area of Budapest
University of Technology and Economics (BME FIKP-MI/SC).  The research
of the second author was supported by National Research, Development
and Innovation Office NKFIH, K-124171.

\end{document}